\documentclass{article}
\usepackage[margin=1in]{geometry}
\usepackage[onehalfspacing]{setspace}
\usepackage{amsmath,mathtools}
\usepackage{pgfplots}
\usepackage{subcaption}
\usepackage{booktabs}
\usepackage[braket, qm]{qcircuit}
\usepackage{graphicx}
\usepackage{authblk}
\usepackage{xcolor}
\usepackage[linesnumbered,ruled,vlined]{algorithm2e}
\usepackage{qcircuit}
\usepackage{booktabs}
\usepackage{cite}

\usepackage{graphicx}
\usepackage{subcaption}

\SetCommentSty{mycommfont}

\SetKwInput{KwInput}{Input}                
\SetKwInput{KwOutput}{Output}  
\SetKwInput{KwParam}{Parameter}            
\SetKwProg{Fn}{Function}{}{}

\newcommand{\meqref}[1]{\text{Eq}.~\eqref{#1}}
\newcommand{\mref}[1]{Sec.\,\,\!$\ref{#1} $}
\newcommand{\mfig}[1]{Fig.\,\,\!$\ref{#1} $}

\usepackage[english]{babel}

\usepackage{pdflscape}
\usepackage{fancyvrb}
\usepackage{calc}
\usepackage{rotating}
\usepackage{pgf,tikz}
\usepackage{mathrsfs}
\usetikzlibrary{arrows}
\usepackage{mathtools}

\usepackage{listings}
\usepackage[]{qcircuit}
\usepackage{braket}
\usepackage{mathtools}
\usepackage{varwidth}
\usepackage{caption}
\usepackage{subcaption}
\usepackage{graphicx}
\usepackage[export]{adjustbox}

\makeatletter
\def\paragraph{\@startsection{paragraph}{4}%
	\z@\z@{-\fontdimen2\font}%
	{\normalfont\bfseries}}
\makeatother

\usepackage{graphicx}
\usepackage{pgffor}
\usepackage{tikz,tikz-cd}
\usetikzlibrary{backgrounds,fit,decorations.pathreplacing}  
\usepackage{booktabs}
\usepackage{enumerate}

\usepackage{amssymb, amsmath, amsthm}

\usepackage{xypic}
\usepackage{scalerel}
\usepackage{ulem}

\usepackage{graphicx}              
\usepackage{amsmath}               
\usepackage{amsfonts}              
\usepackage{amsthm}                
\usepackage{xkeyval}               
\usepackage{amssymb}
\usepackage{enumerate}             
\usepackage{color}
\usepackage{breqn}                 
\usepackage{bm}                    

\allowdisplaybreaks[1]
\usepackage{nicefrac}
\usepackage{booktabs}

\usepackage{kpfonts}
\usepackage{stackengine}
\usepackage{calc}
\newlength\shlength
\newcommand\xshlongvec[2][0]{\setlength\shlength{#1pt}%
	\stackengine{-5.6pt}{$#2$}{\smash{$\kern\shlength%
			\stackengine{7.55pt}{$\mathchar"017E$}%
			{\rule{\widthof{$#2$}}{.57pt}\kern.4pt}{O}{r}{F}{F}{L}\kern-\shlength$}}%
	{O}{c}{F}{T}{S}}

\usepackage{hyperref}
\hypersetup{
	colorlinks   = true,    
	citecolor    = red      
}

\newcommand{\RN}[1]{%
	\textup{\uppercase\expandafter{\romannumeral#1}}%
}

\allowdisplaybreaks

\newtheorem{thm}{Theorem}[subsection]

\newtheorem{lem}[thm]{Lemma}

\newtheorem{cor}[thm]{Corollary}

\newtheorem{example}[thm]{Example} 
\newtheorem{remark}[thm]{Remark}

\def\<{\langle}
\def\>{\rangle}

\numberwithin{equation}{section}

\newcommand{\abs}[1]{\lvert#1\rvert}

\begin{document}


\title{Quantum Detection of Sequency-Band Structure}

\author[1]{Alok Shukla \thanks{Corresponding author.}}
	\author[2]{Prakash Vedula}
	\affil[1]{School of Arts and Sciences, Ahmedabad University, India}
	\affil[1]{alok.shukla@ahduni.edu.in}
	\affil[2]{School of Aerospace and Mechanical Engineering, University of Oklahoma, USA}
	\affil[2]{pvedula@ou.edu}
	

	\date{}
	
	\maketitle

\begin{abstract}

We present a quantum algorithm for estimating the amplitude content of user-specified sequency bands in quantum-encoded signals. The method employs a sequency-ordered Quantum Walsh--Hadamard Transform (QWHT), a comparator-based oracle that coherently marks basis states within an arbitrary sequency range, and Quantum Amplitude Estimation (QAE) to estimate the total probability mass in the selected band. This enables the detection of structured signal components, including both high- and low-sequency features, as well as the identification of rapid sign-change behavior associated with noise or anomalies.
The proposed method can be embedded as a module within a larger quantum algorithm; in this setting, both the input and output remain fully quantum, enabling seamless integration with upstream and downstream quantum operations.
We show that the sequency-ordered QWHT can be implemented with circuit depth $O(\log_2 N)$ (equivalently $O(n)$ for $N=2^n$) when acting on an amplitude-encoded quantum state, whereas computing the full Walsh--Hadamard spectrum of an explicit length-$N$ classical signal requires $O(N\log_2 N)$ operations via the fast Walsh--Hadamard transform. 
This results in an exponential quantum advantage when the QWHT is used as a modular block within a larger quantum algorithm, relative to classical fast Walsh–Hadamard transform–based approaches operating on explicit data.
From an application perspective, the proposed sequency band-energy estimation may be interpreted as a structure-based anomaly indicator, enabling the detection of unexpected high-sequency components relative to a nominal low-sequency signal class. The algorithm is applicable to quantum-enhanced signal processing tasks such as zero-crossing analysis, band-limited noise estimation, and feature extraction in the Walsh basis.
\end{abstract}

\section{Introduction}\label{sec:intro}
In recent years, quantum computing has matured into a powerful algorithmic framework, supported by a growing body of quantum algorithms that demonstrate provable advantages over classical methods for structured computational tasks \cite{nielsen2002quantum}. This includes Deutsch–Jozsa algorithm \cite{deutsch1992rapid}, the Bernstein–Vazirani algorithm \cite{bernstein1993quantum} and its probabilistic generalization \cite{shukla2023generalizationBV}, Simon’s algorithm \cite{simon1997power}, quantum search algorithm \cite{grover1996fast, shukla2025efficientGrover}, quantum Fourier transform , quantum phase estimation \cite{kitaev1995quantum, dobvsivcek2007arbitrary,brassard2002quantum,     shukla2025towardsQPE}, quantum amplitude estimation \cite{brassard2002quantum, giurgica2022low, shukla2025modularQAE} and Shor’s polynomial-time algorithms for integer factorization and discrete logarithms and its variants \cite{shor1999polynomial, shukla2025modularShor}. Quantum algorithms have since been developed for a wide range of scientific and engineering applications, including trajectory optimization \cite{shukla2019trajectory}, the solution of large-scale linear systems \cite{harrow2009quantum}, linear and nonlinear differential equations \cite{childs2020quantum,berry2014high,shukla2023hybriddiffeq}, and digital image, signal processing and spectral analysis \cite{shukla2023hybridimageprocessing,shukla2023quantumDSP,rohida2024hybridPolar,rohida2025quantumEdgeDetect}, efficient state preparation \cite{gleinig2021efficient,  mozafari2020automatic, shukla2024efficientUSP, rosenkranz2025quantum}, quantum optimal control \cite{sandesara2024quantumOptimalcontrol} and numerical integration \cite{shimada2020quantum, abrams1999fast, shukla2025efficientPartialSum}. This paper contributes to this broader landscape by developing a fully quantum, sequency-resolved noise-detection primitive based on the sequency-ordered Quantum Walsh-Hadamard Transform (QWHT) and amplitude estimation.

Taken together, these developments underscore the increasing relevance of quantum algorithms for practical computational tasks, extending well beyond purely complexity-theoretic demonstrations.

Spectral analysis plays a central role in many scientific and engineering disciplines, including signal processing, image analysis, communications, and biomedical diagnostics. Detecting the energy content in specific frequency or sequency bands is critical for identifying structured noise patterns, abrupt transitions, or meaningful signal features. Classical approaches typically rely on transforms such as the Discrete Fourier Transform (DFT) or the Walsh-Hadamard Transform (WHT), followed by band selection and energy estimation. The Walsh basis, particularly in its sequency-ordered form, is closely tied to the number of sign changes (or zero-crossings) in basis functions, making it a natural tool for such analysis.

Quantum computing provides new tools to address these tasks more efficiently in high-dimensional settings. In this work, we propose a quantum algorithm that performs band-specific sequency analysis by leveraging the sequency-ordered Quantum Walsh-Hadamard Transform (QWHT), a custom comparator-based oracle to select basis states in a target sequency band, and amplitude estimation to quantify their contribution. Unlike prior approaches limited to global spectral properties or requiring classical postprocessing, our method enables fully quantum estimation of amplitudes in arbitrary Walsh bands.

A classical technique for assessing signal frequency content is the count of zero-crossings in the signal. Zero-crossings, corresponding to sign changes in sampled data, are a well-known heuristic for detecting high-frequency behavior. Moreover, the Walsh-Hadamard Transform (WHT) in its sequency-ordered form, provides a direct spectral basis in which high-frequency components correspond to high-sequency indices. Quantum computing offers the possibility of analyzing signal structure more efficiently in high-dimensional spaces.

\subsection*{Main Contributions} In this paper, we present a quantum algorithm that addresses the zero-crossing detection problem from a fully quantum perspective. First, we implement the Quantum Walsh-Hadamard Transform in sequency order using a low-depth quantum circuit that avoids classical postprocessing. We then design a comparator-based quantum oracle that efficiently flags those basis states whose sequency indices lie within a user-specified high-frequency band. This is followed by the application of Quantum Amplitude Estimation, which allows us to estimate the probability mass associated with the flagged sequency components. Crucially, this estimated amplitude is shown to be directly related to the number of zero-crossings in the original signal. As a result, the entire procedure constitutes a new quantum algorithm for zero-crossing detection that is both rigorous and resource-efficient.

The proposed algorithm is broadly applicable to several quantum-enhanced tasks, including signal and image processing, biomedical signal analysis (such as EEG and ECG), spectral sparsity detection in compressed sensing, and feature selection in quantum machine learning pipelines. We support our theoretical development with explicit circuit constructions, oracle design details, and numerical simulations on small quantum systems. Altogether, our results demonstrate how spectral properties of classical signals can be extracted efficiently in a quantum computational framework.

In this work, we also present a quantum circuit for performing the Walsh-Hadamard transforms in sequency ordering. We note that usually the Walsh-Hadamard transform in natural order appears in quantum algorithms (for example, Deutsch-Jozsa algorithm \cite{deutsch1992rapid}, Bernstein–Vazirani algorithm \cite{bernstein1993quantum}, Simon's algorithm \cite{simon1997power}, Grover's algorithm \cite{grover1996fast}, Shor's Algorithm \cite{shor1999polynomial}, etc.), often to get a uniform superposition of quantum states at the beginning of the quantum algorithm.
The Walsh-Hadamard transform in natural order and the associated Walsh basis functions in natural order have also been used in the solution of non-linear ordinary differential equations \cite{shukla2023hybriddiffeq}.
On the other hand, the Walsh basis functions and the Walsh-Hadamard transforms in sequency ordering have found applications in several domains in engineering, for example, digital image and signal processing~\cite{kuklinski1983fast, zarowski1985spectral}, cryptography~\cite{lu2016walsh}, solution of non-linear ordinary differential equations and partial differential equations~\cite{beer1981walsh,ahner1988walsh,gnoffo2014global, gnoffo2015unsteady, gnoffo2017solutions}.
In image processing applications, it is often desired to use  Walsh basis functions in sequency ordering because of better energy compaction properties.
For instance, the hybrid classical-quantum algorithm approach for image processing discussed in \cite{shukla2023hybridimageprocessing} uses the Walsh-Hadamard transforms in sequency ordering.
We note that in \cite{shukla2022quantumzerocrossing}, the quantum Walsh-Hadamard transform is performed in natural order, followed by classical computations to obtain the transforms in sequency ordering.
In this and other applications wherein the Walsh-Hadamard transform in sequency ordering is needed, it is desirable to have a quantum circuit to compute the Walsh-Hadamard transforms directly in sequency ordering.

This article is organized as follows. Section~\ref{sec:notation} introduces the notation used throughout the paper. In Section~\ref{Sec:zero-crossing-problem}, we formalize the zero-crossing counting problem. Section~\ref{sec:Sequency-ordered-WH-transforms} presents a quantum circuit for implementing the Walsh--Hadamard transform in sequency ordering. Numerical simulation results validating the proposed framework are reported in Section~\ref{sec:numerical-simulations}. Finally, Section~\ref{sec:conclusion} concludes the paper and discusses future directions.

\subsection{Notation} \label{sec:notation}Before proceeding further, we define some convenient notations used in the rest of the paper.  
	\begin{itemize}
		\item $ \oplus $ : $ x \oplus y $ will denote $ x + y \mod 2 $.
	    \item $ s \cdot x $ : For $ s = s_{n-1}\,s_{n-2}\,\ldots \, s_1\, s_0 $ and $ x =  x_{n-1}\,x_{n-2} \,\ldots \, x_1\, x_0 $ with $ s_i,\, x_i \in \{0,1\}$, $ s \cdot x $ will denote the bit-wise dot product of  $s $ and $ x $ modulo $ 2 $, i.e.,  $s \cdot x :=  s_{0}x_{0} + s_{1}x_{1}+ \ldots + s_{n-1}x_{n-1} \pmod 2 $.
	    \item  $ s(m) $ : For  $ s =  s_{n-1}\,s_{n-2}\,\ldots\, s_2\,s_1\,s_0 $ and $ 1 \leq m \leq n $, $ s(m) $ denotes the string formed by keeping only the $ m $ least significant bits of $ s $, i.e., $ s(m) = s_{m-1}\,s_{m-2}\,\ldots\, s_2\,s_1\,s_0 $. 
	    \item On a few occasions,  by abuse of notations, a non-negative integer $ s $, such that  $ s = \sum_{j=0}^{n-1} \, s_j 2^j $, will be used to represent the $ n $-bit string
	    $ s_{n-1}\,s_{n-2}\,\ldots\, s_2\,s_1\,s_0 $.   
\end{itemize}

\section{Zero-crossings counting problem} \label{Sec:zero-crossing-problem}
Let the function $ {\displaystyle F \colon \{0,1\}^{n}\rightarrow \{1,-1\}} $ be defined as $ F(x) = (-1)^{f(x)}$,  where $ f(x)= s \cdot x $ for some fixed string $ s \in \{0,1\}^{n}  $. The number of zero-crossings (i.e., sign changes) for the sequence 
		\begin{equation}\label{eq:def_sequence}
			\mathcal{S} = 	\left(\, F(0), \, F(1),\, F(2),\, \cdots \cdots ,\, F(2^n-1) \,\right)
		\end{equation}
is defined as 	\begin{equation}\label{Eq:defn_zero_crossings}
	\frac{1}{2} \sum_{k=0}^{N-2} \, \abs{   F(k+1) - F(k)}
	=  \frac{1}{2} \sum_{k=0}^{N-2} \, \abs{   (-1)^{s \cdot (k+1)} - (-1)^{s \cdot k}},
\end{equation}
where $ N =2^n $. 
The numbers of zero-crossings for sequences associated with different secret strings (for $ n=3 $) are listed in Table \ref{tab:zero-crosings-table}.
\begin{table}[]
	\centering
	\begin{tabular}{@{}ccc@{}}
		\toprule
		Secret string, s & Sequence $ \mathcal{S} = ( \, F(k) \, )_{k=0}^{7} $ & Number of zero-crossings (sign changes) \\ \midrule
		$ 000 $	& $ \displaystyle  \left( \, 1, \, 1, \, 1, \, 1, \, 1, \, 1, \, 1, \, 1        \, \right)$  &     $ 0 $                                  \\
		$ 001 $	& $ \displaystyle  \left( \,  1, -1,  1, -1,  1, -1,  1, -1 \, \right) $  &     $ 7 $                                  \\
		$ 010 $	& $ \displaystyle  \left(\,  1,  1, -1, -1,  1,  1, -1, -1 \, \right) $  &      $ 3 $                                \\
		$ 011 $	& $ \displaystyle  \left( \, 1, -1, -1,  1,  1, -1, -1,  1 \, \right) $  &      $ 4 $                                 \\
		$ 100 $	& $ \displaystyle  \left( \, 1,  1,  1,  1, -1, -1, -1, -1 \, \right) $  &      $ 1 $                                 \\
		$ 101 $	& $ \displaystyle  \left( \, 1, -1,  1, -1, -1,  1, -1,  1 \, \right) $  &      $ 6 $                                \\
		$ 110 $	& $ \displaystyle  \left( \, 1,  1, -1, -1, -1, -1,  1,  1 \, \right) $  &      $ 2 $                                 \\
		$ 111 $	& $ \displaystyle  \left( \, 1, -1, -1,  1, -1,  1,  1, -1   \, \right) $  &     $ 5 $                                     \\ \bottomrule
	\end{tabular}
	\caption{The table shows the number of zero-crossings (sign changes) for the sequence $ ( \, F(k) \, )_{k=0}^{7} $, with $ F(k) = (-1)^{s\cdot x} $, for $ s =  000, \, 001,\, \ldots,\, 111 $. }
	\label{tab:zero-crosings-table}
\end{table}
The computation of the number of zero-crossings for the sequence associated with the secret string  $ s=101 $ is illustrated in Ex.\,\,\!$ \ref{ex:one}$.
\begin{example} \label{ex:one}
	Let $ n = 3 $ and $ s = 101 $ (or equivalently $ s = 5$). 
	Clearly, 
	\begin{align*}
		s\cdot 0 &= (101)\cdot (000) = (1 \times 0 ) \, \oplus \,  (0 \times 0 ) \, \oplus \,  (1 \times 0 ) \,  =  0, \\
		s\cdot 1 &= (101)\cdot (001) = (1 \times 0 ) \, \oplus \,  (0 \times 0 ) \, \oplus \,  (1 \times 1 ) \,   =  1, \\
		s\cdot 2 &= (101)\cdot (010) = (1 \times 0 ) \, \oplus \,  (0 \times 1 ) \, \oplus \,  (1 \times 0 ) \,   =  0, \\
		\cdots & \cdots \\
		s\cdot 7 &= (101)\cdot (111) = (1 \times 1) \, \oplus \,   (0 \times 1) \, \oplus \,   (1 \times 1)  \, = 0. 
	\end{align*}
	Then $$ \mathcal{S} = \left((-1)^{s\cdot 0}, \,(-1)^{s\cdot 1},\, (-1)^{s\cdot 2},\, \ldots \,,\,   (-1)^{s\cdot 7} \right) =  \left( \, 1, -1,  1, -1, -1,  1, -1,  1 \, \right) .$$ 
	Let $ N= 2^n = 8 $. The number of zero-crossings of the sequence $ \mathcal{S} $ is given by 
	\[
	\frac{1}{2} \sum_{k=0}^{N-2} \, \abs{   (-1)^{s \cdot (k+1)} - (-1)^{s \cdot k}} = 
	\frac{1}{2} \left( \abs{-1 -1 } + \abs{1 - (-1)} + \abs{-1 - 1} +  \abs{-1 - (-1)} + \abs{1 - (-1)}  + \abs{-1 -1} + \abs{1 - (-1)} \right) = 6.
	\]
\end{example}

\subsection{A technical lemma} \label{sec:Lemma} 
Let 
\begin{equation}\label{eq:g}
	g = \sum_{j=0}^{n-1} \, g_j 2^j,
\end{equation}
and the bits of $ g $ are given by
\begin{align*}
	g_{n-1} & = \, s_0, \\
	g_{n-2} & =  \, s_0 \, \oplus \, s_1, \\ 
	\cdots&  \,\cdots \\
	g_1 &=  \, s_0 \, \oplus \, s_1 \, \oplus \, s_2 \, \oplus \, \ldots  \, \oplus \, s_{n-2}, \\
	g_0 &= \, \, s_0 \, \oplus \, s_1 \, \oplus \, s_2 \, \oplus \, \ldots  \, \oplus \, s_{n-2} \, \oplus \, s_{n-1}.
\end{align*}  
It means,
\begin{align}\label{eq:gk}
	g_{n-1} = s_0,  \quad \text{and} \quad	g_k &= \, \, s_0 \, \oplus \, s_1 \, \oplus \, \ldots   \, \oplus \, s_{n-1-k}, \quad \text{for $ k = n-2 $ to $ 0 $}.
\end{align}
Finally, all the qubits (except the ancilla qubit) are measured and stored in a classical register. The classical bits resulting from these measurements are labeled $ g_0 $ to $ g_{n-1} $ in \mfig{fig:sequency}. Indeed,  $ g $ is number of zero-crossings of the sequence $ \mathcal{S} $ (as defined in \meqref{eq:def_sequence}). 
We recall that, the number of zero-crossings of the sequence $ \mathcal{S} $ (as defined in \meqref{eq:def_sequence}) is given by 
\[
\frac{1}{2} \sum_{k=0}^{N-2} \, \abs{   (-1)^{s \cdot (k+1)} - (-1)^{s \cdot k}},
\]
as noted in  \meqref{Eq:defn_zero_crossings}.
It follows from Corollary \ref{cor} that $ g = Z_n(s) = \frac{1}{2} \sum_{k=0}^{N-2} \, \abs{   (-1)^{s \cdot (k+1)} - (-1)^{s \cdot k}}$ is the number of number of zero-crossings of the sequence $ \mathcal{S} $ that we wanted to determine.

 In this section, we will prove Lemma~\ref{lemma} and its corollary (Corollary \ref{cor}) to show that  $ g = \sum_{j=0}^{n-1} \, g_j 2^j $ (see  \meqref{eq:g}  and  \meqref{eq:gk}) gives the number of of zero-crossings in the sequence $ \mathcal{S} $ given by 
 \[
 \mathcal{S} = 	\left(\, F(0), \, F(1),\, F(2),\, \cdots \cdots ,\, F(2^n-1) \,\right),
 \]
 with $ F(x) = (-1)^{f(x)} $ and $ f(x) = s \cdot x $ (ref.\,\,\!\meqref{eq:def_sequence}). 
 
\begin{lem} \label{lemma}
	For an integer $ m $ with $ 1 < m \leq n $, and for 
	$ x = x_{n-1}\,x_{n-2} \,\ldots \, x_1\, x_0 $ (or equivalently, $ x $ with a decimal representation  $  x = \sum_{j=0}^{m-1} \, x_j 2^j  $),
	 with $ x_j \in \{0,1\} $,   define $ Z_m(x) $ as
	\begin{equation}\label{eq:def_Z_k}
		Z_m(x) := \frac{1}{2} \sum_{k=0}^{2^m-2} \, \abs{   (-1)^{x \cdot (k+1)} - (-1)^{x \cdot k}}.
	\end{equation}
Then, we have 
\begin{equation}\label{eq:lemma}
Z_m(s(m)) = 2 Z_{m-1} (s(m-1)) + (s_0 \, \oplus \, s_1 \, \oplus \, s_2 \, \oplus \, \ldots \, \oplus \,  s_{m-1}).	
\end{equation}
Here $ s(m) = s_{m-1}\,s_{m-2}\,\ldots \, s_1\, s_0 $
and $ s_0 \, \oplus \, s_1 \, \oplus \, s_2 \, \oplus \, \ldots \, \oplus \,  s_{m-1}= (s_0 + s_1 + s_2 + \ldots + s_{m-1}) \pmod 2 $. (ref.\,\,\!\mref{sec:notation}). 
\end{lem}

\begin{proof}
Let $ M = 2^m $. We put $ x=s(m) $ and split the summation on the right side of  \meqref{eq:def_Z_k}  into three different parts as follows.
\begin{align}	\label{eq:lem_all_terms}
	&Z_m(s(m)) \nonumber \\ &= \frac{1}{2} \left(\sum_{k=0}^{\frac{M}{2}-2} \, \abs{   (-1)^{s(m) \cdot (k+1)} - (-1)^{s(m) \cdot k}}\right) + \frac{1}{2} \abs{  (-1)^{s(m) \cdot \left(\frac{M}{2}\right)} - (-1)^{s(m) \cdot \left(\frac{M}{2} -1 \right) }} + \frac{1}{2} \sum_{k=\frac{M}{2}}^{M-2} \, \abs{   (-1)^{s(m) \cdot (k+1)} - (-1)^{s(m) \cdot k}}.
\end{align} 
We note that $ \frac{M}{2} = 2^{m-1}$ represents the $ m $-bit string $ 10\ldots0 $. Therefore, $ s(m) \cdot \left(\frac{M}{2}\right) = s_{m-1} $. Similarly,  $ \frac{M}{2} -1 = 2^{m-1}-1$ represents the $ m $-bit string $ 011\ldots1 $, hence  $ s(m) \cdot \left(\frac{M}{2} -1 \right) =  s_0 \, \oplus \, s_1 \, \oplus \, s_2 \, \oplus \, \ldots \, \oplus \,  s_{m-2}$.
Therefore, the middle term in the above summation reduces to 
\begin{align}\label{eq:lem_second_term}
	\frac{1}{2} \abs{  (-1)^{s(m) \cdot \left(\frac{M}{2}\right)} - (-1)^{s(m) \cdot \left(\frac{M}{2} -1 \right) }} & = \frac{1}{2} \abs{  (-1)^{s_{m-1}} - (-1)^{ s_0 \, \oplus \, s_1 \, \oplus \, \ldots \, \oplus \,  s_{m-2} }} \nonumber \\ & = s_0 \, \oplus \, s_1 \, \oplus \, s_2 \, \oplus \, \ldots \, \oplus \,  s_{m-1}.
\end{align}
Next we show that the last and the first terms are equal. 
\begin{align}\label{eq:lem_last_term}
\frac{1}{2} \sum_{k=\frac{M}{2}}^{M-2} \, \abs{   (-1)^{s(m) \cdot (k+1)} - (-1)^{s(m) \cdot k}} & = 	\frac{1}{2} \sum_{k=0}^{\frac{M}{2}-2} \, \abs{   (-1)^{s(m) \cdot (\frac{M}{2} + k+1)} - (-1)^{s(m) \cdot (\frac{M}{2} + k)}}  \nonumber \\
  & =   \frac{1}{2} \sum_{k=0}^{\frac{M}{2}-2} \, \abs{   (-1)^{s_{m-1}}  \left((-1)^{s(m) \cdot (  k+1)} - (-1)^{s(m) \cdot (k)}  \right)  } \nonumber \\
  & =  \frac{1}{2} \sum_{k=0}^{\frac{M}{2}-2} \, \abs{    \left((-1)^{s(m) \cdot (  k+1)} - (-1)^{s(m) \cdot (k)}  \right)  }. 
\end{align}
From \meqref{eq:lem_all_terms}, \meqref{eq:lem_second_term} and \meqref{eq:lem_last_term} it follows that
\begin{align}
Z_m(s(m)) & = \left( \sum_{k=0}^{\frac{M}{2}-2} \, \abs{   (-1)^{s(m) \cdot (k+1)} - (-1)^{s(m) \cdot k}}\right) + (s_0 \, \oplus \, s_1 \, \oplus \, \ldots \, \oplus \,  s_{m-1}) \nonumber \\
& =  \left( \sum_{k=0}^{\frac{M}{2}-2} \, \abs{   (-1)^{s(m-1) \cdot (k+1)} - (-1)^{s(m-1) \cdot k}}\right) + (s_0 \, \oplus \, s_1 \, \oplus \, \ldots \, \oplus \,  s_{m-1}).  
\end{align}
The last step follows because as $ k $ runs through $ 0$ to  $ \frac{M}{2} -2 = 2^{m-1} -2 $, the computations of  $ s(m) \cdot (k+1) $ and $ s(m) \cdot k $  involve only the $ m-1 $ least significant bits of $ s $, allowing one to write $ s(m) \cdot (k+1) = s(m-1) \cdot (k+1) $ and $ s(m) \cdot k = s(m-1) \cdot k $. 
Hence, we obtain
\begin{align*}
Z_m(s(m)) = 2 Z_{m-1} (s(m-1)) + (s_0 \, \oplus \, s_1 \, \oplus \, s_2 \, \oplus \, \ldots \, \oplus \,  s_{m-1}),
\end{align*}
and the proof is complete.
\end{proof}

\begin{cor}\label{cor}
	Let  $ s = s_{n-1}\,s_{n-2}\,\ldots \, s_1\, s_0 $ with $ s_j \in \{0,1\} $. If $ 	Z_n(s) $ is defined as 
	\[
		Z_n(s) := \frac{1}{2} \sum_{k=0}^{2^n-2} \, \abs{   (-1)^{s \cdot (k+1)} - (-1)^{s \cdot k}},
	\]
	then 
	\begin{equation}\label{eq:cor}
		Z_n(s) = \sum_{k=0}^{n-1} \, g_k 2^k,
	\end{equation}
where $ g_{n-1} = s_0 $ and  $ g_k =  s_0 \, \oplus \, s_1 \, \oplus \, \ldots \,  \oplus \, s_{n-1-k}\,$ for $ k=n-2$ to $ k=0 $.
\end{cor}
\begin{proof}
	It is easy to see from  \meqref{eq:lemma} that 
	\begin{align*}
		Z_1(s(1)) &  = s_0  = g_{n-1}\\ 
		Z_2(s(2)) & = 2 Z_1(s(1)) + (s_0 \, \oplus \,  s_1) = 2 g_{n-1} + g_{n-2} \\ 
		Z_3(s(3)) & = 2 Z_2(s(2)) + (s_0 \, \oplus \,  s_1 \, \oplus s_2) = 2^2 g_{n-1} + 2 g_{n-2} + g_{n-3}. 
	\end{align*}
A simple induction argument, whose details we skip, and the observation that $s(n) = s$  (see \mref{sec:notation}), shows that
\[
	Z_n(s) = \sum_{k=0}^{n-1} \, g_k 2^k.
\]
This completes the proof. 
\end{proof}
Clearly, $ Z_n(s) $ computes the number of zero-crossings of the sequence $ \mathcal{S} $ (ref.\,\,\!\meqref{eq:def_sequence} and  \meqref{Eq:defn_zero_crossings}).
In the following, we give an example to illustrate the steps for computing $ Z_n(s) $.
\begin{example}\label{ex:two}
	Let $ n = 3 $ and $ s = 5 $ (or equivalently $ s = 101 $ as a binary string). 
We have $$ \mathcal{S} = \left((-1)^{s\cdot 0}, \,(-1)^{s\cdot 1},\, (-1)^{s\cdot 2},\, \ldots \,  ,\, (-1)^{s\cdot 7} \right) =  \left( \, 1, -1,  1, -1, -1,  1, -1,  1 \, \right) .$$ 
Let $ N= 2^n = 8 $. 
As $ s = 101 $, we have $ s_0 = 1 $, $ s_1 =0 $ and $ s_2 = 1 $. The number of zero-crossings of the sequence $ \mathcal{S} $ is given by 
\begin{align}\label{eq:eaxample_one}
	Z_3(s) = Z_3(s(3)) = 2 Z_2 (s(2)) + (s_0 \, \oplus \, s_1 \, \oplus \, s_2) = 2 Z_2(s(2)) + (1 \, \oplus \, 0 \, \oplus \, 1) = 2 Z_2(s(2)).
\end{align}
We have, $ s(2) = 01 $. Therefore,
\begin{equation}\label{eq:example_two}
	Z_2(s(2)) = 2 Z_1(s(1)) + (s_0 \, \oplus \, s_1 \,) = 2 Z_1(s_0) + ( 1 \, \oplus \, 0 \,) = 2 (1) + 1 = 3.
\end{equation}
From \meqref{eq:eaxample_one}  and  \meqref{eq:example_two}  we get $ Z_3(s)) = 6 $, which is the same result that we obtained in Ex.\,\,\!$ \ref{ex:one} $. 
\end{example}

\section{Sequency ordered Walsh-Hadamard transforms} \label{sec:Sequency-ordered-WH-transforms}
In this section, we will briefly recall the Walsh basis functions in sequency and natural ordering. We will also discuss the Walsh-Hadamard transforms in sequency and natural ordering and describe a quantum circuit for performing the Walsh-Hadamard transforms in sequency ordering. Interested readers may refer to \cite{beauchamp1975walsh} for further details on Walsh basis functions, Walsh-Hadamard transforms, and their applications. 

\subsection{Walsh basis functions in sequency and natural ordering}
Walsh basis functions $ W_k (x) $ for $  k =0,~1,~2, ~\ldots~ N-1 $  in sequency order are defined as follows
\begin{align}
	W_0(x) &= 1 \quad \text{for } 0 \leq x \leq 1,  \\
	W_{2k} (x) &= W_k(2x) + (-1)^k W_k (2x -1 ),  \\
	W_{2k+1} (x) &= W_k(2x) - (-1)^k W_k (2x -1 ), \\
	W_k(x) &= 0 \quad \text{for } x < 0 \text{ and } x >1,
\end{align}
where $ N $ is an integer of the form $ N = 2^n$. 
For $ N =8 $ the Walsh functions in sequency order are shown in \mfig{fig_walsh_sequency}. 

\begin{figure}[htbp]
	\begin{center}
			\includegraphics[trim=100 225 100 225, clip]{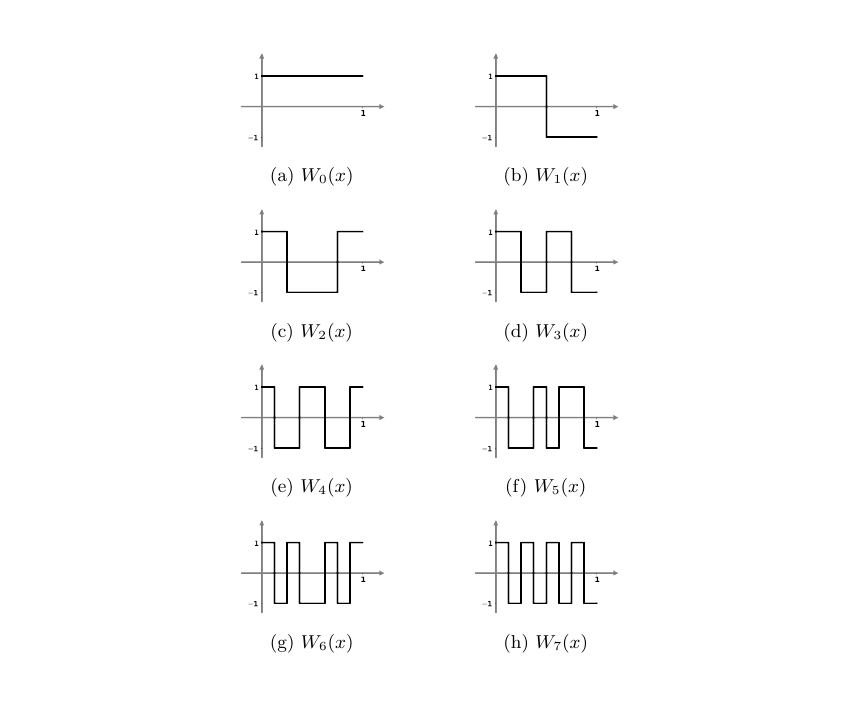}
		\caption{Walsh basis functions in the sequency ordering for $ N=8 $.} 	\label{fig_walsh_sequency}
	\end{center}
\end{figure}

The number of sign changes (or zero-crossings) for the Walsh functions increases as the orders of the functions increase. A vector of length $ N $ can be obtained by sampling a Walsh basis function. The Walsh-Hadamard transform matrix in sequency order is obtained by arranging the vectors obtained from sampling the Walsh basis functions as the rows of a matrix. The vectors are arranged in increasing order of sequency. The Walsh-Hadamard transform matrix of order $ N=8 $ in sequency order is
\begin{align*}
	\frac{1}{\sqrt{8}} \,
	\begin{pmatrix*}[r]
		1 & 1 & 1 & 1 & 1 & 1 & 1 & 1  \\
		1 & 1 & 1 & 1 & -1 & -1 & -1 & -1  \\
		1 & 1 & -1 & -1 & -1 & -1 & 1 & 1  \\
		1 & 1 & -1 & -1 & 1 & 1 & -1 & -1  \\
		1 & -1 & -1 & 1 & 1 & -1 & -1 & 1  \\
		1 & -1 & -1 & 1 & -1 & 1 & 1 & -1  \\
		1 & -1 & 1 & -1 & -1 & 1 & -1 & 1  \\
		1 & -1 & 1 & -1 & 1 & -1 & 1 & -1  \\
	\end{pmatrix*}.
\end{align*} 
In contrast, the Walsh-Hadamard transform matrix of order $ N =2^n $ in natural order is given by  
\begin{align*}
	H_N = H^{\otimes n},
\end{align*}  
where 
\begin{align*}
	H = \frac{1}{\sqrt{2}} \, \begin{pmatrix*}[r]
		1 & 1  \\
		1 & -1  \\
	\end{pmatrix*}.
\end{align*}
and $ N = 2^n $. 
\begin{figure}[htbp]
	\begin{center}
			\includegraphics[trim=100 225 100 225, clip]{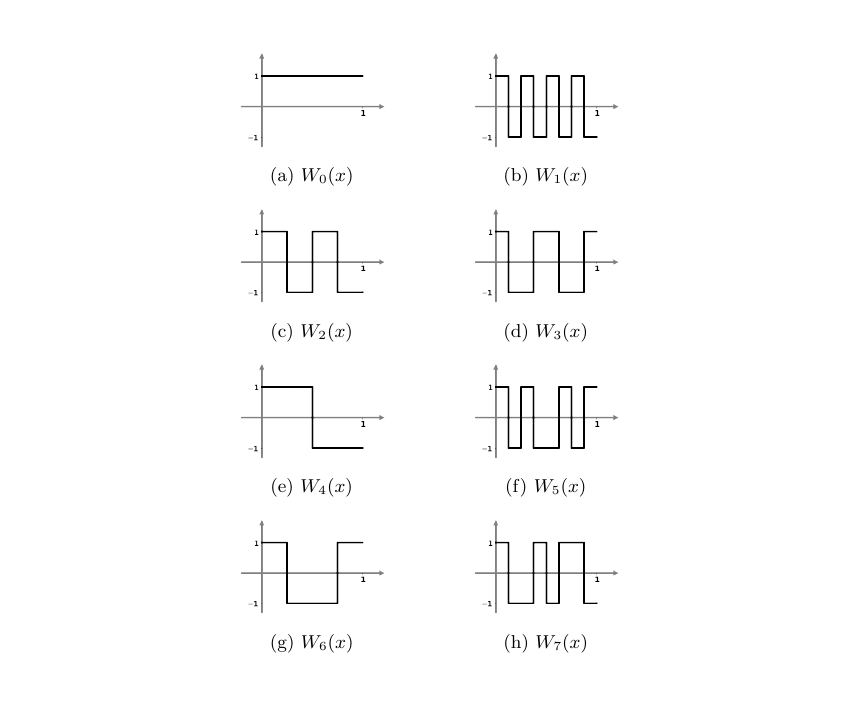}
		\caption{Walsh basis functions in the natural ordering for $ N=8 $.}	\label{fig_walsh_natural}
	\end{center}
\end{figure}
The Walsh-Hadamard matrix in natural order for $ N=8 $ is 
\begin{align*}
		\frac{1}{\sqrt{8}} \,
	\begin{pmatrix*}[r]
		1 & 1 & 1 & 1 & 1 & 1 & 1 & 1  \\
		1 & -1 & 1 & -1 & 1 & -1 & 1 & -1  \\
		1 & 1 & -1 & -1 & 1 & 1 & -1 & -1  \\
		1 & -1 & -1 & 1 & 1 & -1 & -1 & 1  \\
		1 & 1 & 1 & 1 & -1 & -1 & -1 & -1  \\
		1 & -1 & 1 & -1 & -1 & 1 & -1 & 1  \\
		1 & 1 & -1 & -1 & -1 & -1 & 1 & 1  \\
		1 & -1 & -1 & 1 & -1 & 1 & 1 & -1  \\
	\end{pmatrix*}.
\end{align*}
Walsh basis functions (or Hadamard-Walsh basis functions) in natural order can be obtained using the rows of the Walsh-Hadamard matrix $ H_N$. For $ N =8 $, the Walsh functions in natural order are shown in \mfig{fig_walsh_natural}.

The Walsh-Hadamard transforms in natural and sequency orders can also be defined in terms of their actions on the computational basis vectors. 
Let $ N=2^n $ be a positive integer. Let $ V$ be the $ N $ dimensional complex vector space generated by the computational basis states $ \{ \ket{0}, \, \ket{1}, \, \ldots \,,\, \ket{N-1} \} $.  We note that the Walsh-Hadamard transform in natural order can be defined as a linear transformation $ H_N : V \to V $ such that the action of $ H_N = H^{\otimes n}$ on the computational basis state $ \ket{j} $, with $ 0 \leq j \leq N-1 $ is given by
\begin{equation}\label{eq:natual:hadamard}
	H_N \, \ket{j} = \frac{1}{\sqrt{N}} \sum_{k=0}^{N-1} \, (-1)^{k \cdot j} \, \ket{k}. 
\end{equation}
Here $ k \cdot j $ denotes bit-wise dot product of $ k $ and $ j$. 

Next, we note that the Walsh-Hadamard transform in sequency order can be defined as a linear transformation $ H_S : V \to V $ acting on the  basis state $ \ket{j} $, with $ 0 \leq j \leq N-1 $, as follows (see \cite{beauchamp1975walsh}). 
\begin{equation}\label{eq:sequency:hadamard}
	H_S \, \ket{j} = \frac{1}{\sqrt{N}} \sum_{k=0}^{N-1} \, (-1)^{ \sum_{r=0}^{n-1} \, k_{n-1-r} (j_r \oplus j_{r+1}) } \, \ket{k},
\end{equation}
where $ k = k_{n-1}\,k_{n-2}\,\ldots \, k_1\, k_0 $ and $ j =  j_{n-1}\,j_{n-2} \,\ldots \, j_1\, j_0 $, are binary representations of $ k $ and $ j $ respectively, with $ k_i,\, j_i \in \{0,1\} $ for $ i=0,\,1,\, \ldots ,\,n-1$, and $ j_{n} = 0 $.

\subsection{Quantum circuit for performing the sequency ordered Walsh-Hadamard transforms}%

\begin{figure}
	\centering
	\includegraphics{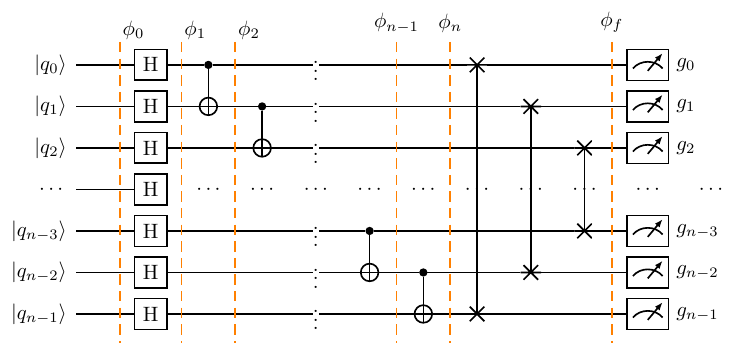}
	\caption{Quantum circuit for computing the Walsh-Hadamard transform in sequency ordering.} \label{fig:sequency}
\end{figure}

A schematic quantum circuit for performing the Walsh-Hadamard transform in sequency order is shown 
In order to show that the quantum circuit in \mfig{fig:sequency} computes the Walsh-Hadamard transforms in sequency ordering, we will compute the output when the input is a computational basis state $ \ket{\phi_0} =  \ket{j} $, with $ 0 \leq j  \leq N-1 $ and $ N =2^n $. 
We have 
\begin{align}
	\ket{\phi_1 } = H^{\otimes n} \ket{\phi_0} &=  \frac{1}{\sqrt{N}}  \sum_{s=0}^{N-1}   \, (-1)^{j \cdot s} \, \ket{s},  \nonumber \\
	&= \frac{1}{\sqrt{N}}  \sum_{s=0}^{N-1}   \,  \, (-1)^{ \sum_{k=0}^{n-1} \, j_k s_k} \, \ket{s}. \label{eq:phif}
\end{align}
Similar to our discussion in \mref{sec:Lemma}, the remaining part of the circuit acts on $ \ket{s} $ to give $ \ket{g} $, where $ g = \sum_{k=0}^{n-1} \, g_k 2^k $ and 
$ g_{n-1} = s_0 $, $  g_k = \, s_0 \, \oplus \, s_1 \, \oplus \, \ldots   \, \oplus \, s_{n-1-k}$, for $ k = n-2, \, \ldots, \, 1,\, 0 $ (see  \meqref{eq:gk}  and \meqref{eq:g}). 
Further, it is easy to see that 
\begin{equation}\label{eq:stog}
	s_k = g_{n-k} \, \oplus \, g_{n-k-1}, \quad k=0,\, 1,\, \ldots,\, n-1,  
\end{equation}
where it is assumed that $ g_n = 0 $. Also, the map sending $ \ket{s} $ to $ \ket{g} $ defined in (see \meqref{eq:gk}  and  \meqref{eq:g}) is invertible. It can be easily checked and it also follows from the fact that the transformation that sends  $ \ket{s} $ to $ \ket{g} $ is implemented by unitary quantum gates. Therefore, the summation that runs through $ s=0 $ to $ s= N-1 $ in   \meqref{eq:phif}  can be replaced by the summation running through $ g =0 $ to $ g=N-1 $. It follows that
\begin{align} \label{eq:phif_final}
	\ket{\phi_f }  &= \frac{1}{\sqrt{N}}  \sum_{g=0}^{N-1}   \, (-1)^{\sum_{k=0}^{n-1} \, j_k (g_{n-k}  \, \oplus \,  g_{n-k-1} )}   \, \ket{g}  \nonumber \\
	&= \frac{1}{\sqrt{N}}  \sum_{g=0}^{N-1}   \, (-1)^{\sum_{r=0}^{n-1} \, j_{n-1-r} (g_{r} \,  \oplus \, g_{r+1} )}   \, \ket{g}. 
\end{align}
From  \meqref{eq:sequency:hadamard} and  \meqref{eq:phif_final}  it follows that the quantum circuit shown in \mfig{fig:sequency} can be used for computing the Walsh-Hadamard transform in sequency order.

\section{Quantum Algorithm for High-Sequency Noise Detection}
\label{sec:algorithm}

In this section, we present a quantum algorithm designed to detect high-sequency noise components in quantum-encoded signals. Such components often correspond to rapid variations or abrupt transitions in the underlying classical signal, and their detection is vital in various applications, including signal processing, communication systems, and biomedical diagnostics. High-sequency noise, represented in the Walsh-Hadamard basis, is characterized by non-negligible amplitudes associated with high-sequency indices. These indices correspond to the number of zero-crossings in the respective Walsh functions. Our objective is to estimate the amplitude contribution of a specified high-sequency band in a given quantum state.

Given a normalized test signal encoded as a quantum state, our method applies a quantum circuit to perform a sequency-ordered Quantum Walsh-Hadamard Transform (QWHT), followed by a band-selective oracle that flags the presence of amplitudes in the target sequency range. 
Finally, depending on the intended use of the algorithm, the flagged
sequency-band information may either be retained coherently for further
quantum processing or used to estimate the corresponding probability
mass via quantum amplitude estimation.

\subsection{Problem Statement}
Let $N = 2^n$ denote the dimension of the Hilbert space for an $n$-qubit system. Given a normalized quantum state $\ket{\psi_{\text{test}}}$ representing a test signal, our goal is to estimate the probability mass $P_{[a,a+M)}$ in a specific high-sequency band defined by integers $a$ and $M$ such that $0 \le a < a+M \le N$. The target probability mass is given by:
\[
P_{[a,a+M)} = \sum_{j=a}^{a+M-1} \left| \langle j | \psi_{\text{sequency}} \rangle \right|^2,
\]
where $\ket{\psi_{\text{sequency}}} = \mathcal{W} \ket{\psi_{\text{test}}}$ denotes the sequency-domain representation obtained by applying the QWHT unitary $\mathcal{W}$.

\subsection{Overall Algorithm Approach}
The algorithm proceeds through three key stages. First, we transform the
input state into the sequency basis by applying the QWHT, thereby aligning
the state representation with Walsh basis functions ordered by sequency.
Second, we apply a comparator-based oracle $U_S$ that coherently marks
basis states whose indices fall within the desired high-sequency band.
Third, depending on a classical control flag, the algorithm either
retains the resulting flagged quantum state for use as a subroutine
within a larger quantum computation or applies quantum amplitude
estimation to obtain a classical estimate of the probability mass
associated with the flagged sequency band.

\subsection{Detailed Oracle Design ($U_S$)}
\label{sec:oracle_design}

The oracle $U_S$ marks computational basis states $\ket{j}$ lying in the sequency range $[a, a+M)$. This marking uses two quantum comparator circuits: one to test whether $j \ge a$, and another to test whether $j < a+M$. Each comparator uses a flag ancilla qubit to store the result of the inequality.

To check the lower bound, we apply a quantum comparator $C_{\ge}(Q, a)$ which flips a temporary ancilla $Q_{\text{temp1}}$ to $\ket{1}$ if the index $j \ge a$. For the upper bound, we use $C_{<}(Q, a+M)$ to flip another ancilla $Q_{\text{temp2}}$ to $\ket{1}$ if $j < a+M$. These two conditions are then combined via a Toffoli gate with $Q_{\text{temp1}}$ and $Q_{\text{temp2}}$ as controls and a primary flag qubit $Q_{\text{flag}}$ as the target. This marks the main ancilla as $\ket{1}$ if and only if both conditions are satisfied.

To preserve unitarity, the comparators are then uncomputed, restoring $Q_{\text{temp1}}$ and $Q_{\text{temp2}}$ to $\ket{0}$. The oracle thus acts as:
\[
U_S \ket{i}\ket{0}_{\text{flag}} = \begin{cases} \ket{i}\ket{1}_{\text{flag}} & \text{if } a \le i < a+M \\ \ket{i}\ket{0}_{\text{flag}} & \text{otherwise} \end{cases}
\]
This operation coherently encodes membership in the target sequency band
into the flag qubit and may be used either as a control for subsequent
quantum operations or as the input to a quantum amplitude estimation
routine. In Algorithm~\ref{alg:noise_detection}, this behavior is controlled by the
classical Boolean flag \textsc{Estimate}.

\subsection{Amplitude Estimation for Probability Mass}
This step is executed only when a classical estimate of the sequency-band
probability mass is required.
After applying the oracle $U_S$ to the sequency-transformed state $\ket{\psi_{\text{sequency}}}\ket{0}_{\text{flag}}$, the overall system evolves as:
\begin{align*}
U_S \ket{\psi_{\text{sequency}}}\ket{0}_{\text{flag}} &= \sum_{j=a}^{a+M-1} \langle j | \psi_{\text{sequency}} \rangle \ket{j}\ket{1}_{\text{flag}} 
+ \sum_{j \notin [a, a+M)} \langle j | \psi_{\text{sequency}} \rangle \ket{j}\ket{0}_{\text{flag}}.
\end{align*}

The probability of measuring the flag qubit $Q_{\text{flag}}$ in state $\ket{1}$ is equal to the sum of squared amplitudes in the target band:
\begin{equation}
    \Pr(Q_{\text{flag}}=1) = \sum_{j=a}^{a+M-1} \left| \langle j | \psi_{\text{sequency}} \rangle \right|^2 = P_{[a,a+M)}.
\end{equation}

Using QAE, we estimate this quantity efficiently. The final output of the algorithm is an estimate $P_{\text{est}} \approx P_{[a,a+M)}$, representing the total probability mass of the specified high-sequency components in the input signal.

\begin{algorithm}[H]
\SetAlgoLined
\SetNlSty{}{}{}
\SetAlgoNlRelativeSize{-1}
\SetInd{0.2em}{0.4em}
\DontPrintSemicolon

\KwInput{
A normalized quantum state $\ket{\psi_{\text{test}}}$ on $n$ qubits, stored in quantum register $Q$;
Integer $a \in [0,2^n-1]$ (starting sequency index);
Integer $M \in [1,2^n-a]$ (number of consecutive sequency states);
Boolean flag \textsc{Estimate} $\in \{\texttt{true},\texttt{false}\}$.
}

\KwOutput{
If \textsc{Estimate} $=\texttt{false}$: flag qubit $Q_{\text{flag}}$ marking sequency-band membership;\\
If \textsc{Estimate} $=\texttt{true}$: classical estimate $P_{\text{est}}$ of the probability mass in $[a,a+M)$.
}

\tcc{Step 1: State preparation}
Initialize the $n$-qubit register $Q$ in state $\ket{\psi_{\text{test}}}$.

\tcc{Step 2: Sequency-basis transformation via QWHT}
Apply the sequency-ordered Quantum Walsh--Hadamard Transform (QWHT):
\begin{enumerate}[(i)]
\item Apply Hadamard gates $H$ to each qubit in $Q$.
\item For $k=0$ to $n-2$, apply a CNOT gate from $Q_k$ to $Q_{k+1}$, with control being $Q_k$ and target $Q_{k+1}$.
\item Apply SWAP gates to reverse qubit order for correct sequency indexing.
\end{enumerate}

\tcc{Step 3: Sequency-band oracle $U_S$ for $[a,a+M)$}
Introduce ancilla qubits:
\begin{itemize}
\item Flag qubit $Q_{\text{flag}}$ initialized to $\ket{0}$,
\item Temporary ancillas $Q_{\text{temp1}}$ and $Q_{\text{temp2}}$ initialized to $\ket{0}$.
\end{itemize}

Implement $U_S$ as follows:
\begin{enumerate}[(i)]
\item Apply comparator $C_{\ge}(Q,a)$, setting $Q_{\text{temp1}}=1$ iff $Q \ge a$.
\item Apply comparator $C_{<}(Q,a+M)$, setting $Q_{\text{temp2}}=1$ iff $Q < a+M$.
\item Apply a Toffoli gate with controls $(Q_{\text{temp1}},Q_{\text{temp2}})$ and target $Q_{\text{flag}}$.
\item Uncompute $Q_{\text{temp2}}$ using $C_{<}^\dagger(Q,a+M)$.
\item Uncompute $Q_{\text{temp1}}$ using $C_{\ge}^\dagger(Q,a)$.
\end{enumerate}

\noindent
Effect: (Ignoring temporary ancillas)
\[
U_S \ket{i}\ket{0}_{\text{flag}} =
\begin{cases}
\ket{i}\ket{1}_{\text{flag}}, & i \in [a,a+M),\\
\ket{i}\ket{0}_{\text{flag}}, & \text{otherwise}.
\end{cases}
\]

\If{\textsc{Estimate} $=\texttt{false}$}{
\KwRet{Flag qubit $Q_{\text{flag}}$ (coherent sequency-band indicator).}
}
\Else{
\tcc{Step 4: Amplitude estimation (conditional)}
Apply a suitable amplitude estimation algorithm (e.g., MLQAE~\cite{suzuki2020amplitude}
or AWQAE~\cite{shukla2025modularQAE}) using:
State preparation from Step~2, Oracle $U_S$ from Step~3,
and Flag qubit $Q_{\text{flag}}$.

Let $P_{\text{est}}$ denote the estimated probability
$\Pr(Q_{\text{flag}}=1)$.

\KwRet{$P_{\text{est}}$.}
}
\caption{Quantum algorithm for the detection of sequency-band structure}
\label{alg:noise_detection}
\end{algorithm}

\begin{remark}
\label{rem:optional-qae}
Step~4 (Amplitude Estimation) is optional and is included only when a
classical estimate of the sequency-band probability mass is required.
If the proposed procedure is used as a coherent subroutine within a
larger quantum algorithm, the output of Step~3 may be retained in
quantum form without performing amplitude estimation or measurement.
In this case, the flag qubit $Q_{\mathrm{flag}}$ coherently encodes
membership in the sequency band $[a,a+M)$ and can be used as a control
or input for subsequent quantum operations, such as conditional
unitaries, amplitude amplification, or further oracle calls.
\end{remark}

\subsection{Computational Complexity}

We now analyze the computational resources required by the proposed algorithm.
Throughout this discussion, $N = 2^n$ denotes the size of the signal, while $n$
is the number of qubits used to represent it.

The sequency-ordered Quantum Walsh--Hadamard Transform (QWHT) is implemented by
applying a standard $n$-qubit Hadamard transform followed by a cascade of CNOT
and SWAP gates that map the natural ordering to sequency ordering. This
construction requires $O(n)$ Hadamard gates, $O(n)$ CNOT gates, and $O(n)$ SWAP
gates, resulting in an overall circuit depth of $O(n)$ for the QWHT stage.

The oracle $U_S$ consists of two integer comparator circuits acting on
$n$-qubit registers, each requiring $O(n)$ quantum gates and $O(n)$ ancilla
qubits. A Toffoli gate is used to combine the comparator outputs, and the
ancilla registers are subsequently uncomputed to preserve reversibility.
Consequently, the oracle has circuit depth $O(n)$ and requires $O(n)$ ancillary
qubits.

Classically, detection of the sequency band may be
performed deterministically by computing the full discrete Walsh--Hadamard
Transform (DWHT) of a length-$N$ signal. Using fast Walsh--Hadamard transform
algorithms, this requires $O(N \log N) = O(n  2^n)$ time complexity.

When the quantum input state corresponding to a length-$N$ signal can be
prepared efficiently using $O(\mathrm{poly}(n))$ quantum operations, the
proposed algorithm offers a significant advantage for high-dimensional signals
($N >> 1$) and moderate accuracy requirements. We emphasize that the claimed
advantage applies to probabilistic estimation of sequency-band energy and
relies on the assumption of efficient quantum state preparation.

\section{Numerical Simulations}
\label{sec:numerical-simulations}

This section presents numerical simulations validating the correctness of the proposed sequency-ordered Quantum Walsh--Hadamard Transform (QWHT) and the associated sequency band-energy estimation procedure. All simulations are performed using noiseless statevector evolution for small system sizes. The objective of these simulations is to verify consistency with the theoretical analysis.

Unless stated otherwise, simulations are carried out for $n=3$ data qubits, corresponding to $N = 2^n = 8$ computational basis states. This system size is sufficient to capture the essential features of the algorithm while allowing exact classical verification of all results. We consider several representative input signals, including smooth,
low-sequency signals,  highly oscillatory and structured-noise cases, to
illustrate the behavior of the algorithm across different regimes.

\paragraph*{Simulation methodology.}

For a real-valued input signal $\mathbf{x} = (x_0, x_1, \ldots, x_{N-1})$, we prepare the normalized quantum state $\ket{\psi_{\mathrm{in}}} = \sum_{j=0}^{N-1} x_j \ket{j}$, with $\sum_j |x_j|^2 = 1$. Applying the sequency-ordered QWHT yields the state $\ket{\psi_{\mathrm{seq}}} = U_{\mathrm{QWHT}} \ket{\psi_{\mathrm{in}}} = \sum_{k=0}^{N-1} c_k \ket{k}$.

Due to the explicit Gray-code mapping and bit-reversal operations implemented in the circuit, the computational basis states $\ket{k}$ at the output of the QWHT correspond directly to sequency indices. The quantity $|c_k|^2$ therefore represents the contribution of sequency index $k$ to the total signal energy.

To validate the oracle-based band selection, the sequency indices are partitioned into three disjoint bands $[0,b)$, $[b,b+M)$, and $[b+M,N)$. For the simulations that follow we have used $b=2$ and $M=3$. For any band $[u,v)$, the associated band energy is $P_{[u,v)} = \sum_{k=u}^{v-1} |c_k|^2$. This probability mass is exactly the quantity targeted by amplitude estimation in the full algorithm and is extracted directly from the statevector in the present simulations. 

\paragraph*{Constant (DC) input signal.}

Figure~\ref{fig:dc-signal} shows the numerical results for a constant input signal. The time-domain plot confirms the absence of sign changes. The sequency-domain spectrum exhibits a single nonzero component at the lowest sequency index, and the band-energy histogram shows that all probability mass is concentrated in the lowest sequency band.

\begin{figure}[t]
\centering

\begin{subfigure}[t]{0.5\linewidth}
\centering
\includegraphics[width=\linewidth]{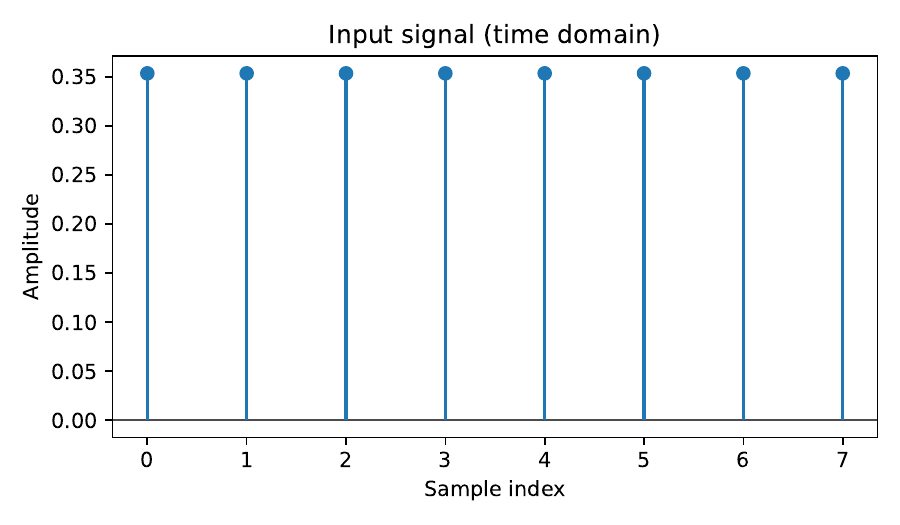}
\caption{Time-domain signal}
\end{subfigure}

\vspace{0.8em}

\begin{subfigure}[t]{0.49\linewidth}
\centering
\includegraphics[width=\linewidth]{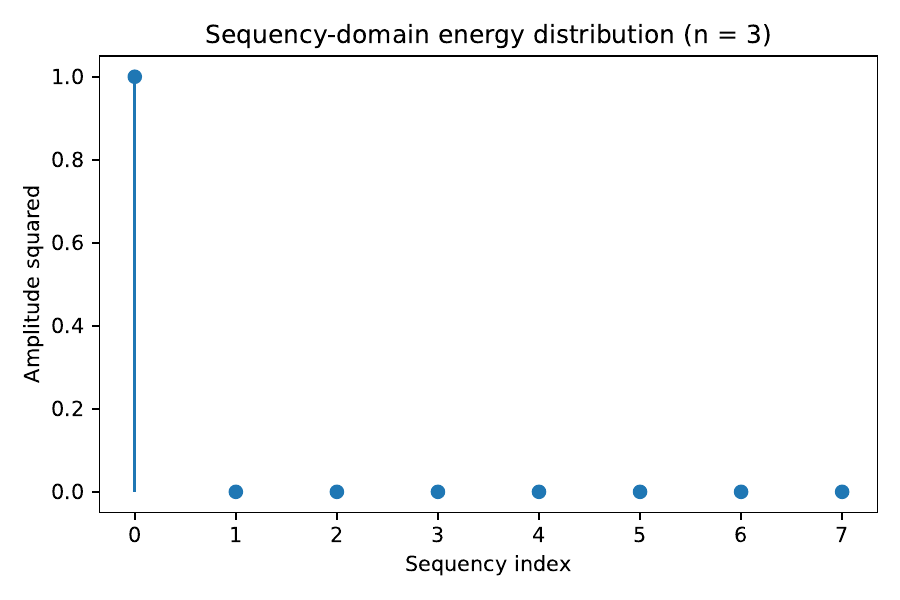}
\caption{Sequency spectrum}
\end{subfigure}
\hfill
\begin{subfigure}[t]{0.49\linewidth}
\centering
\includegraphics[width=\linewidth]{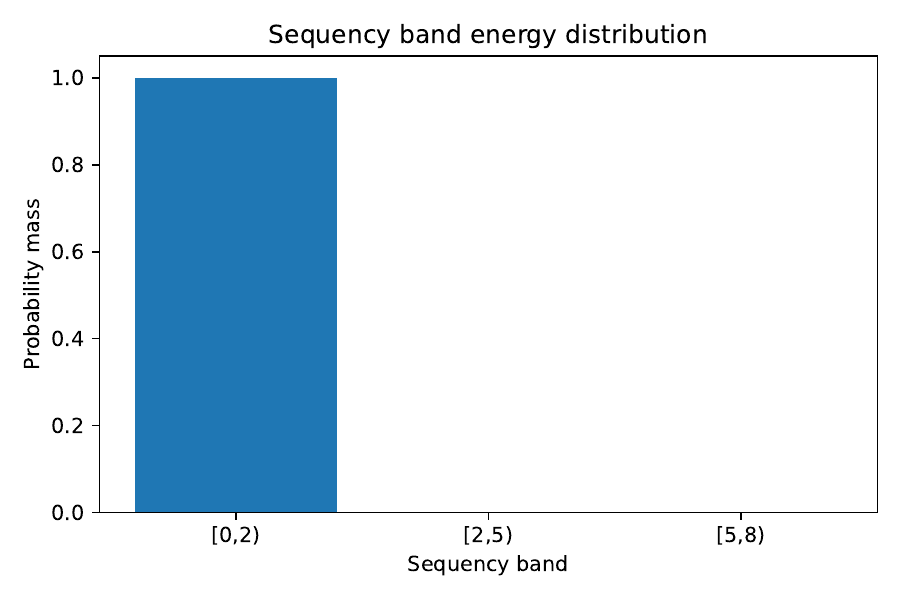}
\caption{Band-energy histogram}
\end{subfigure}

\caption{Numerical results for the constant (DC) input signal.}
\label{fig:dc-signal}
\end{figure}

\paragraph*{Piecewise-constant (edge-like) input signal.}

Figure~\ref{fig:edge-signal} presents the results for a piecewise-constant signal containing a single sharp transition. The time-domain plot highlights one sign change. The sequency-domain spectrum shows energy distributed over low and intermediate sequency indices, while the band-energy histogram confirms that the probability mass is shared between the corresponding sequency bands.

\begin{figure}[t]
\centering

\begin{subfigure}[t]{0.5\linewidth}
\centering
\includegraphics[width=\linewidth]{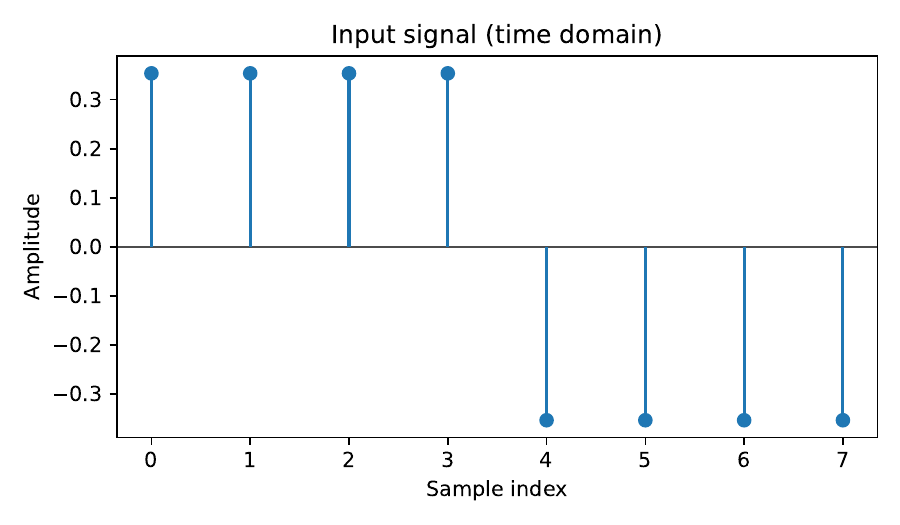}
\caption{Time-domain signal}
\end{subfigure}

\vspace{0.8em}

\begin{subfigure}[t]{0.49\linewidth}
\centering
\includegraphics[width=\linewidth]{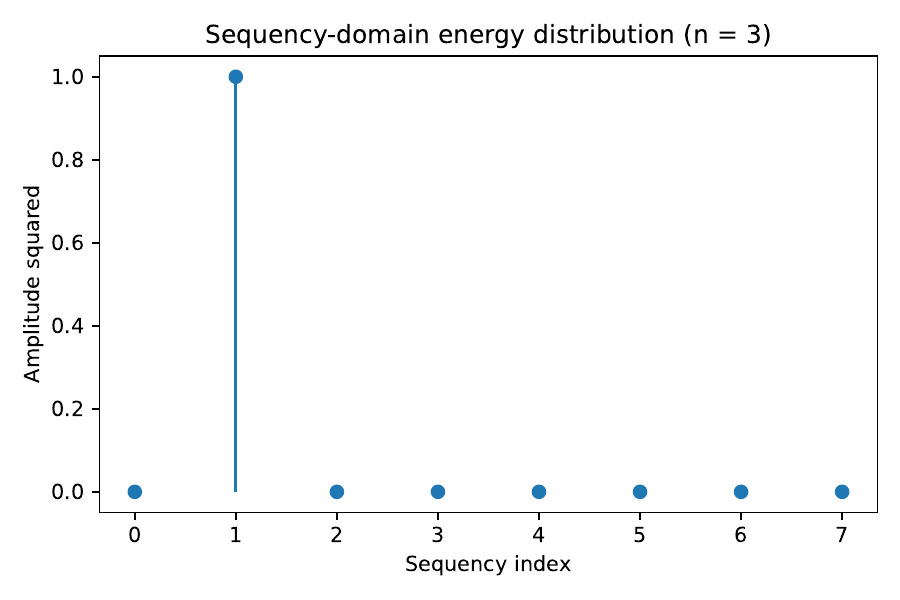}
\caption{Sequency spectrum}
\end{subfigure}
\hfill
\begin{subfigure}[t]{0.49\linewidth}
\centering
\includegraphics[width=\linewidth]{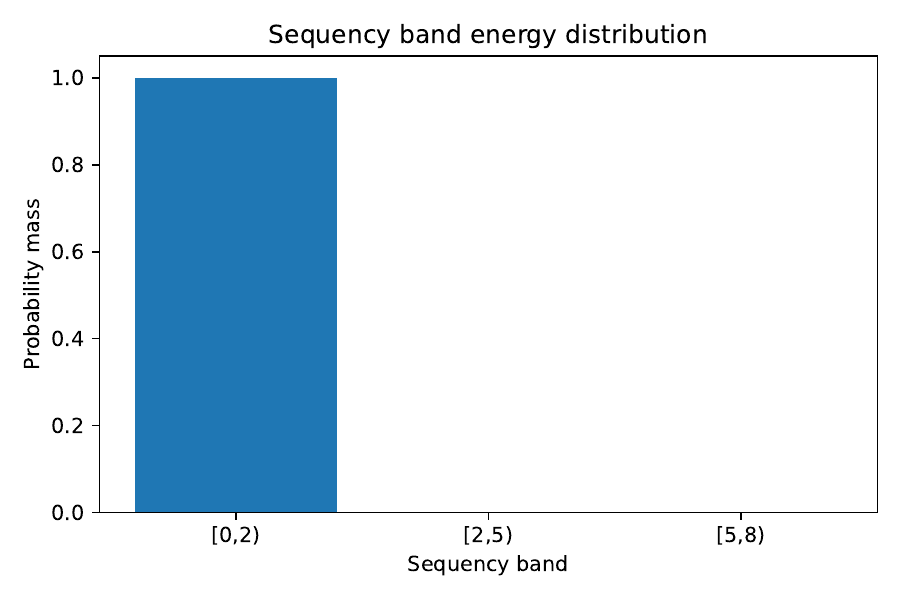}
\caption{Band-energy histogram}
\end{subfigure}

\caption{Numerical results for the piecewise-constant (edge-like) input signal.}
\label{fig:edge-signal}
\end{figure}

\paragraph*{Alternating-sign input signal.}

Figure~\ref{fig:alternating-signal} shows the results for an alternating-sign input signal. The time-domain plot exhibits the maximal number of sign changes. Correspondingly, the sequency-domain spectrum is dominated by high-sequency components, and the band-energy histogram shows that most of the probability mass lies in the high-sequency band.

\begin{figure}[t]
\centering

\begin{subfigure}[t]{0.5\linewidth}
\centering
\includegraphics[width=\linewidth]{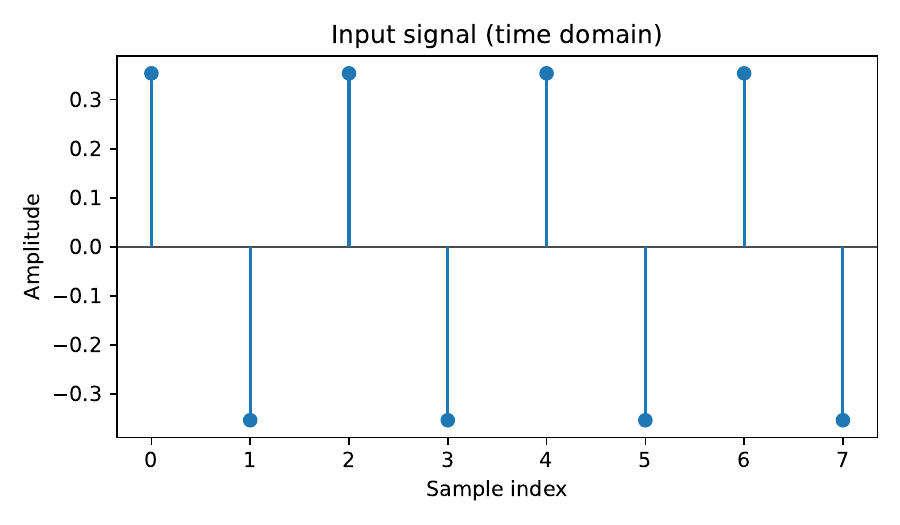}
\caption{Time-domain signal}
\end{subfigure}

\vspace{0.8em}

\begin{subfigure}[t]{0.49\linewidth}
\centering
\includegraphics[width=\linewidth]{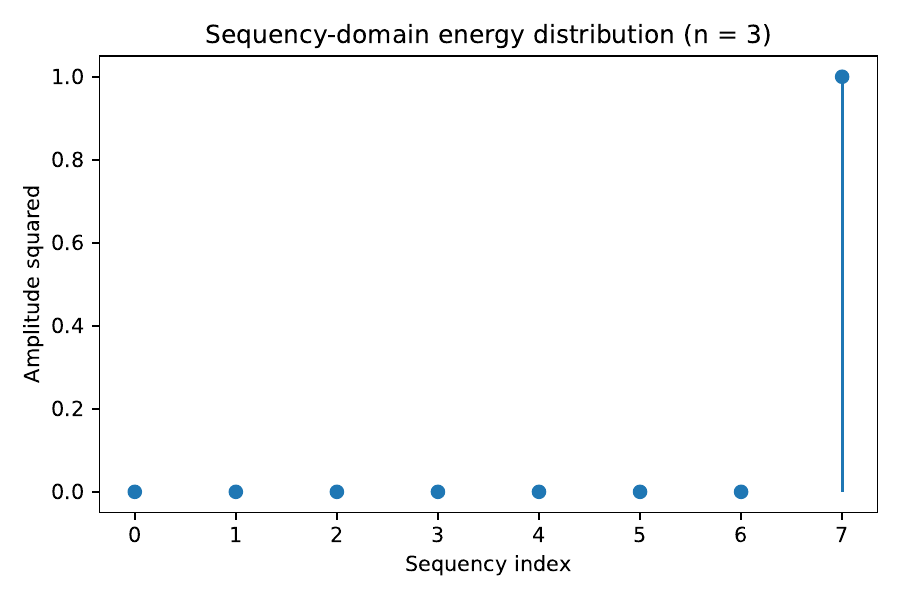}
\caption{Sequency spectrum}
\end{subfigure}
\hfill
\begin{subfigure}[t]{0.49\linewidth}
\centering
\includegraphics[width=\linewidth]{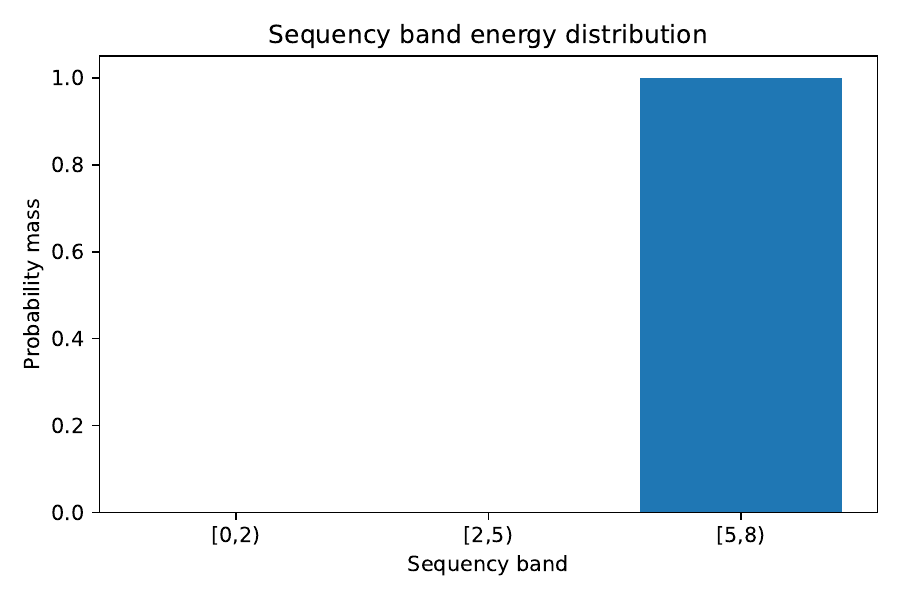}
\caption{Band-energy histogram}
\end{subfigure}

\caption{Numerical results for the alternating-sign input signal.}
\label{fig:alternating-signal}
\end{figure}

\paragraph*{Summary and Discussion.}

Across all tested signals, the numerical simulations show a consistent correspondence between time-domain structure, sequency-domain energy spectra, and sequency band-energy distributions. These results confirm that the proposed QWHT correctly implements sequency ordering, that computational basis states at the output can be interpreted directly as sequency indices, and that the oracle-based procedure accurately estimates aggregate sequency band energy.

Beyond its role in estimating aggregate sequency content, the proposed band-energy estimation procedure admits a natural interpretation as a form of structure-based anomaly detection in the sequency domain. For a given class of nominal signals characterized by predominantly low-sequency structure, the appearance of significant probability mass in higher sequency bands indicates a deviation from the expected signal behavior. In this sense, high-sequency noise or rapid sign-change structure manifests as an anomaly that can be detected by monitoring the estimated band energies. We emphasize that the present work does not address learning-based, adaptive, or statistical anomaly detection, but rather provides a deterministic, oracle-based metric that is directly tied to the physical structure of the signal.

\section{Conclusion}
\label{sec:conclusion}

In this work, we introduced a fully quantum framework for sequency-resolved
signal analysis based on the sequency-ordered Quantum Walsh--Hadamard Transform
(QWHT). By combining a low-depth implementation of the QWHT with a
comparator-based, coherently reversible sequency-band oracle, we showed how to
isolate and quantify the contribution of user-specified sequency bands in
quantum-encoded signals. When a classical estimate is required, the marked
sequency information can be converted into a probability estimate using Quantum
Amplitude Estimation (QAE); alternatively, the output may be retained coherently
and used directly as a subroutine within a larger quantum algorithm.

A key outcome of this work is the demonstration that the sequency-ordered QWHT
can be implemented with circuit depth $O(\log_2 N)$ (equivalently $O(n)$ for
$N=2^n$) on an amplitude-encoded quantum state. In contrast, classical approaches
that explicitly compute the Walsh--Hadamard spectrum of a length-$N$ signal
require $O(N\log_2 N)$ operations. When the proposed method is embedded as a
module within a larger quantum computation, so that both input and output remain
quantum, this represents an exponential advantage.

From an application perspective, sequency-band energy estimation provides a
structure-based indicator of rapid sign-change behavior, closely related to
classical notions of zero-crossings and high-frequency noise. This makes the
proposed algorithm relevant to quantum-enhanced signal processing tasks such as
band-limited noise detection, edge and anomaly detection, and feature extraction
in the Walsh basis. The numerical simulations presented for small system sizes
confirm the correctness of the sequency ordering, the interpretation of output
basis states as sequency indices, and the correspondence between time-domain
structure and sequency-domain energy concentration.

Several directions remain open for future work. These include experimental
realization on near-term quantum hardware, optimization of comparator and oracle
circuits for noise-resilient implementations, and integration with variational or
hybrid quantum–classical schemes for adaptive noise suppression and signal
classification. More broadly, the techniques developed here, particularly
sequency-ordered transforms and band-selective quantum oracles, provide a
reusable primitive for structured spectral analysis within the growing
landscape of quantum signal processing.


\end{document}